\documentclass[sigconf]{acmart}
\usepackage{array}
\usepackage{amsthm}
\usepackage{mathtools}
\usepackage{comment} 
\usepackage{amsmath,amsfonts}
\usepackage{subcaption,multirow,cleveref}
\newtheorem{dfn}{Definition}
\newtheorem{theorem}{Theorem}

\newtheorem{lemma}[theorem]{Lemma}
\crefname{algocf}{alg.}{algs.}
\Crefname{algocf}{Algorithm}{Algorithms}

\def\lc{\left\lceil}   
\def\rc{\right\rceil}
\usepackage[ruled,linesnumbered]{algorithm2e}
\newcommand{\nonl}{\renewcommand{\nl}{\let\nl\oldnl}}
\makeatletter
\newcommand{\nosemic}{\renewcommand{\@endalgocfline}{\relax}}
\renewcommand\footnotetextcopyrightpermission[1]{}

\AtBeginDocument{%
  \providecommand\BibTeX{{%
    \normalfont B\kern-0.5em{\scshape i\kern-0.25em b}\kern-0.8em\TeX}}}

\setcopyright{none}
\copyrightyear{2018}
\acmYear{2018}
\acmDOI{10.1145/1122445.1122456}

\acmConference[Woodstock '18]{Woodstock '18: ACM Symposium on Neural
  Gaze Detection}{June 03--05, 2018}{Woodstock, NY}
\acmBooktitle{Woodstock '18: ACM Symposium on Neural Gaze Detection,
  June 03--05, 2018, Woodstock, NY}
\acmPrice{15.00}
\acmISBN{978-1-4503-XXXX-X/18/06}

\begin{document}

\title{Locality-Sensitive Experience Replay for Online Recommendation}

\author{Xiaocong Chen}
\email{xiaocong.chen@unsw.edu.au}
\affiliation{%
  \institution{University of New South Wales}
  \city{Sydney}
  \country{Australia}
}

\author{Lina Yao}
\email{lina.yao@unsw.edu.au}
\affiliation{%
  \institution{University of New South Wales}
  \city{Sydney}
  \country{Australia}
}

\author{Xianzhi Wang}
\email{xianzhi.wang@uts.edu.au}
\affiliation{%
  \institution{University of Technology Sydney}
  \city{Sydney}
  \country{Australia}
}

\author{Julian McAuley}
\email{jmcauley@eng.ucsd.edu}
\affiliation{%
  \institution{University of California, San Diego}
  \state{CA}
  \country{USA}
}

\renewcommand{\shortauthors}{X.Chen, et al.}

\begin{abstract}
Online recommendation requires handling rapidly changing user preferences. Deep reinforcement learning (DRL) 
is an effective means of
capturing users' dynamic interest during interactions with recommender systems.
Generally, it is challenging to train a DRL agent, due to large state space (e.g., user-item rating matrix and user profiles), action space (e.g., candidate items), and sparse rewards.
Existing studies leverage experience replay (ER) to let an agent learn from past experience. However, they adapt poorly to the complex environment of online recommender systems and are inefficient in determining an optimal strategy from past experience.
To address these issues, we design a novel state-aware experience replay model, which selectively selects the most relevant, salient experiences, and 
recommends the agent with the optimal policy for online recommendation.
In particular, the model uses locality-sensitive hashing to map high dimensional data into low-dimensional representations and a prioritized reward-driven strategy to replay more valuable experience at a higher chance.   
Experiments on three online simulation platforms demonstrate our model's feasibility and superiority to several existing experience replay methods. 
\end{abstract}





\maketitle

\section{Introduction}
Online recommendation aims to
learn 
users'
preferences and recommend items dynamically to help users find 
desired 
items in highly dynamic environments~\cite{zhang2019deep}.
Deep reinforcement learning (DRL) naturally fits online recommendation as it learns policies through interactions with the environment via maximizing a cumulative reward.
Besides, DRL has been widely applied to sequential decision-making (e.g. in~Atari~\cite{mnih2013playing} and AlphaGo~\cite{silver2016mastering}) and achieved remarkable progress. Therefore, it is increasing applied for enhancing online recommender systems~\cite{chen2018stabilizing,bai2019model,zhao2019deep}.

DRL-based recommender systems 
cover three categories of methods:
deep Q-learning (DQN), policy gradient, and hybrid methods.
DQN aims to find the best step via maximizing a Q-value over all possible actions. As the representatives, \citet{zheng2018drn} 
introduced DRL into recommender systems for news recommendation; \citet{chen2018stabilizing} introduced a robust reward function to Q-learning, which stabilized the reward in online recommendation.
Despite the capability of
fast-indexing in selecting a discrete action, Q-learning-based methods conduct the ``maximize" operation over the action space (i.e., all available items) and suffer from the 
\textit{stuck agent problem}~\cite{dulac2015deep}---the ``maximize" operation becomes unfeasible when the action space has high dimensionality (e.g., 100,000 items form a 10k-dimensional action space)~\cite{chen2019large}.
Policy-gradient-based methods use the average reward as guideline to mitigate the 
stuck agent
problem~\cite{chen2019large}. However, they
are prone to converge to sub-optimality~\cite{pan2019policy}. 
While 
both DQN and policy gradient are more suitable for small action and state spaces~\cite{lillicrap2015continuous,wang2018supervised} in a recommendation context, hybrid methods~\cite{chen2019large,dulac2015deep, zhao2018deep, hu2018reinforcement} has the capability to map large high-dimensional discrete state spaces into low-dimensional continuous spaces via combines the advantages of Q-learning and policy gradient.
A typical hybrid method is the actor-critic network~\cite{konda2000actor}, which adopts policy gradient on an actor network and Q-learning on a critic network to achieve Nash equilibrium on both networks. Actor-critic networks have been widely applied to DRL-based recommender systems~\cite{liu2020end,chen2020knowledge}.  

Existing DRL-based recommendation methods except policy-gradient-based ones 
rely heavily
on experience replay to learn from previous experience, avoid re-traversal of the state-action space, and stabilize the training on large, sparse state and action spaces~\cite{zha2019experience}.
They generally require long training time, thus
suffering from the training inefficiency problem.
Further more, in contrast to the larger, diverse pool of continuous actions required in recommendation tasks, existing experience replay methods are mostly designed for games with a small pool of discrete actions. Therefore, a straightforward application of those methods may result in strong biases during the policy learning process~\cite{hou2017novel}, thus impeding the generalization of optimal recommendation results.
For example,
\citet{schaul2015prioritized} assume that not every experience is worth replaying and
propose a
prioritized experience replay (PER) method to replay only the experience with the largest temporal difference error.
\citet{sun2020attentive} propose attentive experience replay (AER), which introduces similarity measurement into PER to boost the efficiency of finding similar states' experience, but attention mechanisms cause inefficiency on large sized state and action spaces~\cite{kitaev2020reformer}.

We present a novel experience replay structure, Locality-Sensitive Experience Replay (LSER), to address the above challenges.
Differing from existing approaches, which apply random or uniform sampling,
LSER samples experiences based on expected states. Inspired by collaborative filtering (which measures the similarity between users and items to make recommendations) and AER~\cite{sun2020attentive}, LSER only
replays experience from 
similar states to improve the sampling efficiency. 
Specifically, we introduce a ranking mechanism to prioritize replays and promote the higher reward experiences.  
We further use $\epsilon$-greedy method to avoid replaying high-rears states excessively.
\begin{figure*}
    \centering
    \includegraphics[width=\linewidth]{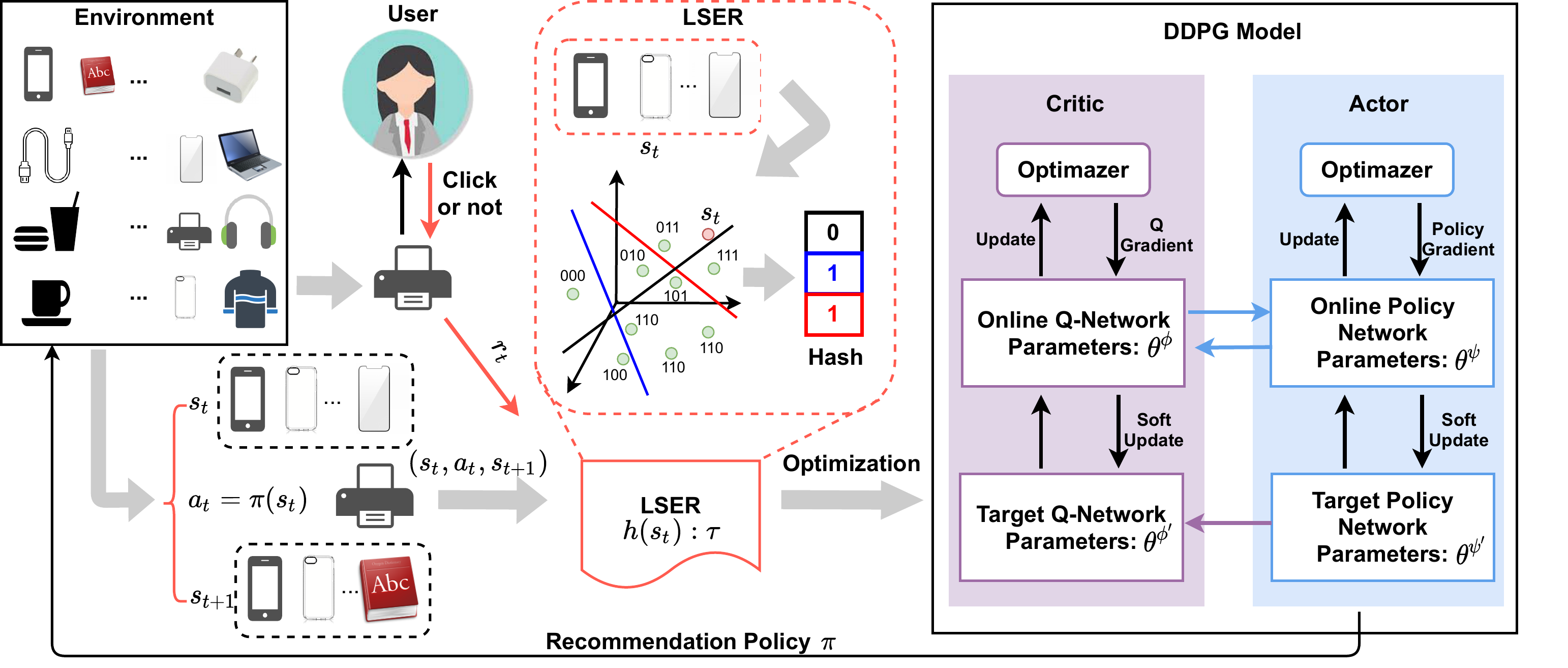}
    \caption{The proposed LSER with DDPG. The environment provides the current state $s_t$, with the policy $\pi$ learned by the DDPG model; the action $a_t$ can be obtained by $a_t = \pi(s_t)$. $r_t$ will be provided by the user 
    (e.g.~click or not).
    LSER takes $s_t$ as the input and encodes it on the projective space. Given the encoded states, LSER will return the most similar experience for DDPG to update he parameters.
    After that, this transition $h(s_t): (s_t, a_t,s_{t+1},r_t)$ will be stored.}
    \label{fig:structure}
\end{figure*}

Considering  the high-dimensionality of vectorized representations of states,
We convert similarity measurement for high-dimensional data into a hash key matching problem and employ \textit{locality-sensitive hashing} to transform states into
low-dimensional representations. Then, we assign similar vectors the same hash codes (based on the property of locality-sensitive hashing).
Such a transformation reduces all the states into low dimension hash keys.

In summary,
we make the following contributions in this paper:
\begin{itemize}
  \item We propose a novel experience replay method (LSER) for reinforcement-learning-based online recommendation. It employs a similarity measurement to improve 
  training efficiency.
  \item LSER replays experience based on the similarity level of the given state and the stored states;
  the agent thus has a higher chance to learn valuable information than it does with 
  uniform
  sampling.
  \item The experiments on three platforms, VirtualTB, RecSim and RecoGym, demonstrate the efficacy and superiority of LSER to several state-of-the-art experience replay methods.
\end{itemize}
\section{Methodology}
In this section, we will briefly introduce the proposed LSER method with theoretical analysis. The overall structure of using LSER in DRL RS can be found in~\Cref{fig:structure}.
\subsection{Overview}
{\em Online Recommendation} aims to find a solution that best reflects real-time interactions between users and the recommender system and apply the solution to the recommendation policy. The system needs to analyze users' behaviors and update the recommend policy dynamically.
In particular, reinforcement learning-based recommendation learns from interactions through a Markov Decision Process (MDP).

Given a recommendation problem consisting of a set of users $\mathcal{U} = \{u_0,u_1,\cdots u_n\}$, a set of items $\mathcal{I} = \{i_0,i_1,\cdots i_m\}$ and user's demographic information $\mathcal{D}=\{d_0,d_1,\cdots,d_n\}$, MDP can be represented as a tuple $(\mathcal{S},\mathcal{A},\mathcal{P},\mathcal{R},\gamma)$, where $\mathcal{S}$ denotes the state space (i.e., the combination of the subsets of $\mathcal{I}$ and its corresponding user information) $\mathcal{A}$ denotes the {\em action space}, which represents agent's selection during recommendation based on the {\em state space} $\mathcal{S}$, $\mathcal{P}$ denotes the set of transition probabilities for state transfer based on the action received, $\mathcal{R}$ is a set of rewards received from users, which are used to evaluate the action taken by the recommender system (each reward is a binary value to indicate whether user has clicked the recommended item or not), and $\gamma$ is a discount factor $\gamma \in [0,1]$ for the trade-off between future and current rewards.

Given a user $u$ and an initial state $s_0$ observed by the agent (or the recommender system), which includes a subset of item set $\mathcal{I}$ and user's profile information $d_0$, a typical recommendation iteration for the user goes as follows:
first, the agent takes an action $a_0$ based on the recommend policy $\pi_0$ under the observed state $s_0$ and receives the corresponding reward $r_0$---the reward $r_0$ is the numerical representation for user's behavior such as click through or not;
then, the agent generates a new policy $\pi_1$ based on the received reward $r_0$ and determines the new state $s_1$ based on the probability distribution $p(s_\textit{new}|s_0,a_0)\in\mathcal{P}$.
The cumulative reward (denoted by $r_c$) after $k$ iterations from the initial state is as follows:
\begin{align*}
  r_c = \sum_{k=0} \gamma^{k}r_k
\end{align*}

DRL-based recommender systems uses a \textit{replay buffer} to store and replay old experience for training. Given the large state and action space in a recommender system, not every experiences are worth to replay~\cite{chen2021survey}---replaying experience that does not contain useful information will increase the training time significantly and introduce extra uncertainty to convergence. Hence, it is reasonable to prioritize replay important experience for DRL recommender systems.

The ideal criterion for measuring the importance of a transition in RL is the amount of knowledge learnable from the transition in its current state ~\cite{schaul2015prioritized,horgan2018distributed}.
State-of-the-art methods like AER are unsuitable for recommendation tasks that contain large, higher dimensional state and action spaces as their sampling strategies may not work properly.
Thus, we propose a new experience replay method named \textit{Locality-sensitive experience replay (LSER)}  for online recommendation, which uses hashing for dimension reduction when sampling and storing the experiences.

\subsection{Locality-sensitive Experience Replay}
We formulate the
storage and sampling issue in LSER as a similarity measure problem, where LSER stores similar states into the same buckets and samples similar experiences based on state similarities.
A popular way of searching similar high-dimensional vectors in 
Euclidean space is Locality-Sensitive Hashing (LSH),
which follows the idea of Approximate Nearest Neighbor (ANN) while allocating similar items into the same buckets to measure the similarity.
However, 
standard LSH conducts bit-sampling on the Hamming space; it requires time-consuming transformation between the Euclidean space to the Hamming space, liable to lose information.
Aiming at measuring the similarity between high-dimensional vectors without losing significant information, we 
propose using $p$-stable distribution~\cite{nolan2003stable} to conduct dimensionality reduction while preserving the original distance. This converts  high-dimensional vectors (states) into low-dimensional representations easier to be handled by the similarity measure.

To address possible hash collision (i.e., dissimilar features may be assigned into the same bucket and recognized as similar),
we introduce the formal definition of the collision probability for LSH. Then, we theoretically analyze the collision probability for $p$-stable distribution to prove that our method has a reasonable boundary for collision probability. 

\begin{dfn}[Collision probability for LSH in $p$-stable distribution]
Given an LSH function $h_{ab}\in\mathcal{H}$ and the probability density function (PDF) of the absolute value of the $p$-stable ($p \in [1,2])$ distribution $f_p(t)$ in $L^p$ space, the collision probability for vectors $\textbf{u}$ and $\textbf{v}$ is represented by:
\begin{align}
    P = Pr[h_{\textbf{ab}}(\textbf{u})=h_{\textbf{ab}}(\textbf{v})] = \int_{0}^w\frac{1}{c}f_p\bigg(\frac{t}{c}\bigg)\bigg(1-\frac{t}{w}\bigg)~dt \label{eq1}
\end{align}
where $c = \|\textbf{u}-\textbf{v}\|_p$ and $w$ is a user-defined fixed distance measure.
\end{dfn}

Here, we use a 2-state distribution, i.e.,~normal distribution for dimension reduction.
We randomly initialize $n_h$ hyperplanes based on normal distribution on the projective space $\mathcal{P}^n$ to get the hash representation for a given state $s$,
where $n$ is the dimension
of the state. The hashing representation $h(s)$ for the given state $s$ is calculated as follows:
\begin{align}
    h_{p\in\mathcal{P}^n}(s) = \{0,1\}^n \text{with} \begin{cases} 
      1 & p_i\cdot s_i >0\\ 
      0 & p_i\cdot s_i \leq 0 ~\label{eq2}
   \end{cases}
\end{align}

The collision probability of the above method can be represented as:
\begin{align}
    & P = \mathit{Pr}[h_{p\in\mathcal{P}^n}(\textbf{u})=h_{p\in\mathcal{P}^n}(\textbf{v})]= 1- \frac{Ang(\textbf{u},\textbf{v})}{\pi} \notag \\ 
    & \text{ where } Ang({\textbf{u},\textbf{v}}) = \arccos\frac{|\textbf{u} \cap \textbf{v}|}{\sqrt{|\textbf{u}|\cdot|\textbf{v}|}} ~\label{eq3}
\end{align}

Eq.(\ref{eq2}) formulates the information loss during the projection, where we use term $e$ to represent the quantification between the real value $p\cdot v$ and hashed results induced from $h(\textbf{v})$.
Since the relative positions in original space are preserved during the hash transformation with an extra measurement $e$, the upper bound and lower bound of collision probability boundary in projective space is guarantee to be intact. That means the more dissimilar states will not receive a higher probability to be allocated into the same hash result.
\begin{lemma}
Given an arbitrary hash function $h_{ab}\in\mathcal{H}$, the collision probability for a given vector $\textbf{u}$ and $\textbf{v}$ is bounded at both ends.
\end{lemma}
\begin{proof}
Since $\mathit{Pr}[h_{ab}(\textbf{u}) = h_{ab}(\textbf{v})]$ monotonically decreases in $c$ for any hash function from the LSH family $\mathcal{H}$, the collision probability is bounded from above by $Pr[h_{ab}(\textbf{u}) = h_{ab}(\textbf{v})]$ for $c-e$ and from below by $Pr[h_{ab}(\textbf{u}) = h_{ab}(\textbf{v})]$ for $c+e$.
\begin{align*}
   P=\int_{0}^w\frac{1}{c}f_p\bigg(\frac{t}{c}\bigg)\bigg(1-\frac{t}{w}\bigg)~dt 
    = \int_{0}^{w/c}f_p(q)\bigg(1-\frac{qc}{w}\bigg)~dq \text{ with } q=\frac{t}{c}
\end{align*}

Then, we have the upper bound:
\begin{align*}
    & \int_0^{w/(c-e)}f_p(q)\bigg(1-\frac{(c-e)q}{w}\bigg)dq \text{ with } q=\frac{t}{c-e}\\
  =& \int_0^{w/(c-e)}\bigg(f_p(q)\bigg(1-\frac{qc}{w}\bigg)+\frac{qef_p(q)}{w}\bigg)~dq \\
  \leq & P + \frac{e}{w}\int_0^{w/(c-e)} qf_p(q)~dq 
  \leq  P + \frac{e}{c-e}
\end{align*}
and the lower bound:
\begin{align*}
    & \int_0^{w/(c+e)}f_p(q)\bigg(1-\frac{(c+e)q}{w}\bigg)~dq \text{ with } q=\frac{t}{c+e}\\
    & = \int_0^{w/(c+e)}\bigg(f_p(q)\bigg(1-\frac{qc}{w}\bigg)-\frac{qef_p(q)}{w}\bigg)~dq \\
    & =P -\frac{e}{w}\int_0^{w/(c+e)} qf_p(q)~dq - \int_{w/(c+e)}^{w/c} f_p(q)\bigg(1-\frac{qc}{w}\bigg)~dq \\
  & \geq P- \frac{e}{c+e} - \bigg(1-\frac{c}{c+e}\bigg) 
  = P - \frac{2e}{c+e}
\end{align*}

We compute the upper bound based on H\"{o}lder's inequality in $L^1$ space:
\begin{align*}
    \int_0^{w/(c-e)} qf_p(q)~dq \leq \bigg(\sup_{q\in [0,w/(c-e)]} q\bigg)\|f_p\|_{1} \leq \frac{w}{c-e}
\end{align*}
Considering the $L^\infty$ space, we have:
\begin{align*}
    \int_0^{w/(c-e)} qf_p(q)~dq \leq \|f_p\|_{\infty} \int_0^{w/(c-e)} q~dq= \frac{w^2\|f_p\|_{{\infty}}}{2(c-e)^2}
\end{align*}
We use the similar method in $L^1$ to compute the lower bound:
\begin{align*}
    \int_{w/(c+e)}^{w/c} f_p(q)\bigg(1-\frac{qc}{w}\bigg)~dq 
    & \leq \bigg(\sup_{q\in [w/(c+e),w/c]} \bigg(1-\frac{qc}{w}\bigg)\bigg)\|f_p\|_{1} \\
    &\leq 1-\frac{cw/(c+e)}{w} = \frac{e}{c+e} 
\end{align*}
and in $L^\infty$:
\begin{align*}
    \int_{w/(c+e)}^{w/c} f_p(q)\bigg(1-\frac{qc}{w}\bigg)~dq
    & \leq \|f_p\|_{\infty} \int_{w/(c+e)}^{w/c} \bigg(1-\frac{c}{w}q\bigg)~dq \\
    & \leq \frac{e^2w\|f_p\|_{\infty}}{2c(c+e)^2}
\end{align*}

The collision probability $Pr[h_{ab}(u) = h_{ab}(v)]$ is bounded from both ends as follows:
\begin{align*}
    \bigg[P - \min\bigg(\frac{2e}{c+e}, \frac{e^2w\|f_p\|_{\infty}}{2(c+e)^2}\bigg), P + \min\bigg(\frac{e}{c-e},\frac{w^2\|f_p\|_{{\infty}}}{2(c-e)^2}\bigg)\bigg]
\end{align*}
\end{proof}

Note that, when calculating the lower and upper bounds, $q$ represents $\frac{t}{c+e}$ and $\frac{t}{c-e}$, respectively.
The  algorithm of LSER is shown in ~\Cref{alg:lsh}.


In the following, we demonstrate from two perspectives that LSER can find the similar states efficiency.
First, we show the efficacy of LSER with theoretical guarantee, i.e., similar states can be sampled given the current state. We formulate `the sampling of similar states' as a neighbor-finding problem in the projective space and provide a theoretical proof of
the soundness of LSER.
Given a set of states $\mathcal{S}$, 
and a query $q_s$, LSER can quickly find a state $s\in\mathcal{S}$ within distance $r_2$ or determine that $\mathcal{S}$ has no states within distance $r_1$.
Based on 
existing work~\cite{indyk1998approximate}, the LSH family 
is 
$(r_1,r_2,p_1,p_2)$-sensitive, i.e., we can
find a distribution $\mathcal{H}$ such that $p_1 \geq Pr_{h\sim \mathcal{H}}[h_{ab}(\textbf{u}) = h_{ab}(\textbf{v})]$ when $\textbf{u}$ and $\textbf{v}$ are similar and $p_2 \leq Pr_{h\sim \mathcal{H}}[h_{ab}(\textbf{u}) = h_{ab}(\textbf{v})]$ when $\textbf{u}$ and $\textbf{v}$ are dissimilar.
\begin{theorem}

Let $\mathcal{H}$
be
$(r_1,r_2,p_1,p_2)$-sensitive.
Suppose $p_1 > 1/n$ and $p_2 > 1/n$, where $n$ is the size of data points. There exists a solution for the neighbor finding problem in LSER within $O(n^\rho p_1^{-1} \log n)$ query time, 
and
$O(n^{1+\rho}p_1^{-1})$ space.
\end{theorem}

\begin{proof}
Assume $r_1,r_2,p_1,p_2$ are known, $\rho = \frac{\log (1/p_1)}{\log(1/p_2)}$, and $k=\frac{\log(n)}{\log(1/p_2)}$ where $k$ is the number of hash functions, and LSH initializes $L$ tables. Based on the definition in~\cite{indyk1998approximate}, we have:
\begin{align}
    kL = k\lc p_{1}^{-k} \rc \leq k(e^{\log (1/p_1) \cdot k} + 1) \leq k(n^\rho/p_1 + 1) = O(n^\rho/p_1\log n)
\end{align}

The space complexity is calculated as $O(Lnd_s)$ where $d_s$ is the dimension of state $s$. It can be written as $O(n^{1+\rho}/p_1d_s)$ (by applying $L = n^\rho/p_1$) and further simplified into $O(n^{1+\rho}/p_1)$.

Then, we prove LSER can find similar neighbors. The $L$ table can be classified into two categories: similar and dissimilar. Given a state $s$, the similar category gives similar states while the dissimilar category provides dissimilar states. We split the two categories such that $L = \lfloor n \rfloor + \lc m \rc$ and its corresponding $\lfloor k \rfloor, \lc k \rc$. Given any state $s\in\mathcal{S}$ in the distance $r_1$, LSER must be able to find the most similar states in a high probability---the query and the data need to share the same hash-bucket in one of the tables. The probability of their not sharing the same hash-bucket is
\begin{align}
    (1-p_1^{\lfloor k \rfloor})^{\lfloor n \rfloor}(1-p_1^{\lc k \rc})^{\lc m \rc} 
    & \leq (1-p_1^{\lfloor k \rfloor})^{n-1}(1-p_1^{\lc k \rc})^m \\
    & \leq e^{-np_1^{\lfloor k \rfloor} - mp_1^{\lc k \rc}}(1-p_1^{\lfloor k \rfloor})^{-1} \\
    & = e^{-(np_1^{-1+\alpha} + mp_1^{\alpha})n^{-\rho}}(1-p_1^{\lfloor k \rfloor})^{-1} \\
    & = e^{-1}(1-p_1^{\lfloor k \rfloor})^{-1}
\end{align}
where $\alpha = \lc k \rc - k$. We have applied the definitions $p_1^k = p_2^{\rho k} = n^{-\rho}$ for step 6 to step 7 and $np_1^{-1+\alpha} + mp_1^{\alpha} = n^{\rho}$ for step (7) to (8). Finally, we get the probability of LSER's getting the similar states as follows:
\begin{align*}
    P \geq 1 - e^{-1}(1-p_1^{\lfloor k \rfloor})^{-1} > 0
\end{align*}

Recall that $p_1 > 1/n$. Therefore, we conclude that LSER can find the most similar states.
\end{proof}

\begin{figure}[ht]
    \centering
    \includegraphics[width=0.8\linewidth]{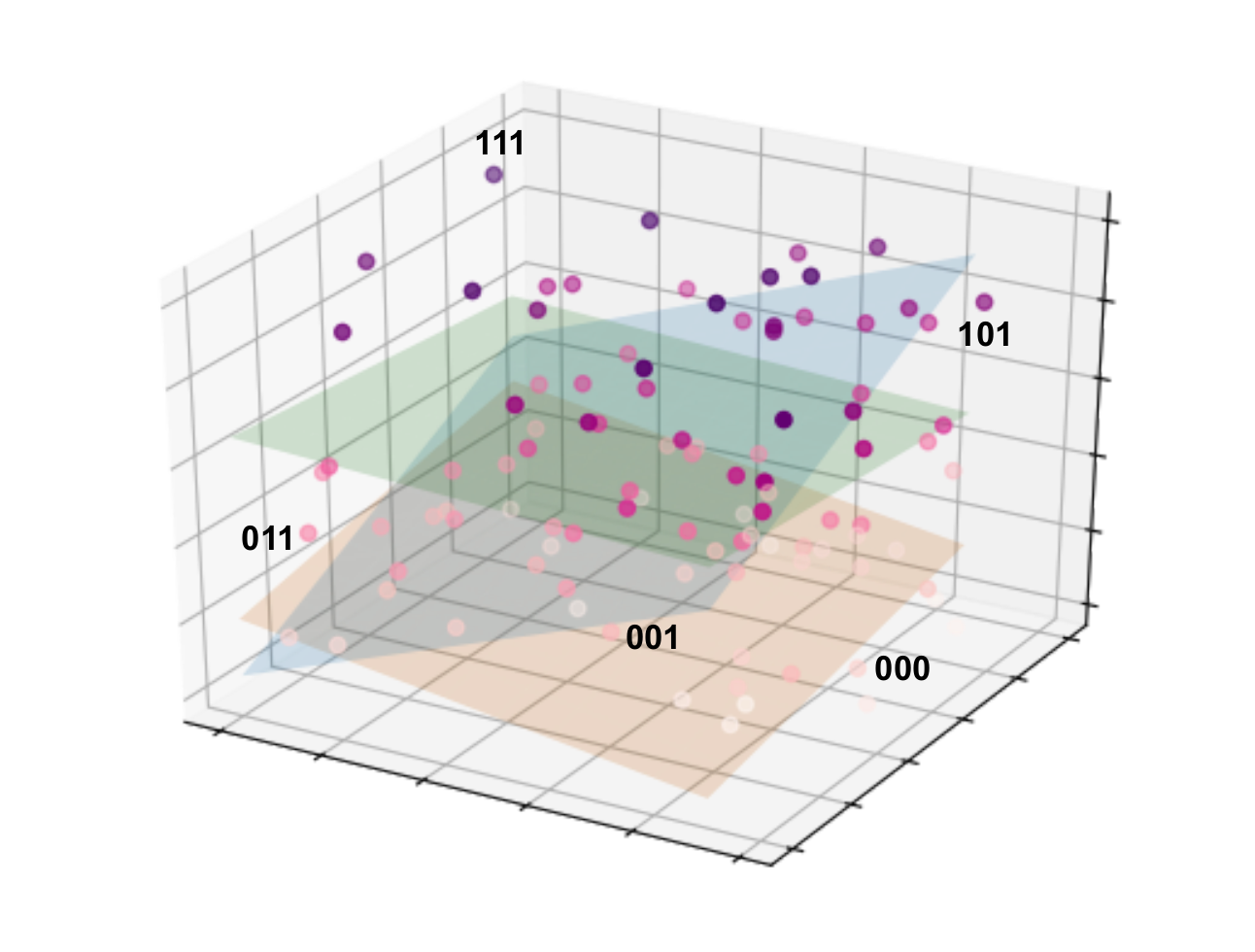}
    \caption{Given a high dimensional space, three random hype-planes are initialized based on normal distribution. Each hype-plane splits the space into two hash areas 0 and 1. The space is split into six hash areas. We can find that, states are encoded into a binary string e.g.,\{111, 101, 011, 001, 000\} }
    \label{fig:lsh}
\end{figure}

\subsection{Store and Sampling Strategy}

Existing experience replay methods in DRL research assume that the 
recent
experience is 
more informative
than older experience. Therefore, they simply
replace the oldest experience with the newest experience 
to
update the experience buffer in DRL-based recomender systems without further optimization.
As such, some valuable experience might be discarded, i.e.,~catastrophic forgetting.
In contrast, we design a state-aware reward-driven experience storage strategy, which removes the experience with the lowest reward---instead of following the First-In-First-Out (FIFO) strategy---when the replay buffer is full.
Formally speaking, a transition $\tau_t:(s_t,a_t,s_{t+1},r_t)$ will be stored 
in
the replay buffer based on the value $h_{p\in\mathcal{P}^n}(\tau_t.s_t)$. If the replay buffer is full, 
the transition with the same value of $h_{p\in\mathcal{P}^n}(\tau_t.s_t)$ but lower reward will be replaced.
In practice, an indicator $m_t$ is stored in the transition as well to indicate when the recommendation should terminate.

Sampling strategy is another crucial component of LSER, which determines which experience should be selected for the agent to optimize in LSER.
We propose a state-aware reward-driven sampling strategy that only replays the experience with the top-N highest rewards in the same hashing area; this way, the agent can quickly find the correct direction for optimization.
We call our sampling strategy `state-aware' because we use a \textit{hash key} to encode the state and replay the experience based on the hash key.
Compared with uniform sampling, our strategy has a higher chance to replay the correct experience.
Here, we illustrate how to address three related challenges faced by our sampling strategy: \textit{exploitation-vs-exploration dilemma}, \textit{bias annealing} and \textit{non-existence dilemma}.

\vspace{3mm}
\noindent{\em Exploitation vs. exploration dilemma.} Exploitation and exploration dilemma is a well-known dilemma when training an agent for RL, including LSER.
T
While our reward-driven strategy forces 
the
agent 
to 
exploit
existing high-rewarding experiences,
the agent may converge to a sub-optimal policy instead of the globally optimal one.
We use a similar method to $\epsilon$-greedy
to achieve a trade-off between exploitation and exploration.
LSER first draws a random probability $p\in[0,1]$ then uses reward-driven sampling if the probability less than a threshold, $\epsilon_{\mathit{max}}$ and random sampling otherwise. The threshold allows LSER to replay low priority experience to fulfill the exploration requirement. 

\vspace{3mm}
\noindent{\em Bias annealing. }
Prioritizing
partial experiences among the replay buffer may introduce 
inductive bias~\cite{schaul2015prioritized,sun2020attentive}---the training process is highly non-stationary (due to 
changing policies); even a small bias introduced by the sampling strategy may change the solution that the policy converges to.
A
common solution is to let
the priority anneal periodically so that the agent can visit those less-replayed experiences.
By using the threshold, our $\epsilon$-greedy method has the similar 
effect as annealing on allowing low-priority experiences to be replayed.

\vspace{3mm}
\noindent{\em Non-existence dilemma. }
When split the projective space into areas to initialize hyperplanes, some areas may not have any data points (esp.~when the number of hyperplanes is large), causing the `non-existence dilemma'.
Consequently, when a new transition comes, the algorithm will stop if 
no experience can be found on $h_{p\in\mathcal{P}^n}$.
We use the similarity measure to overcome this problem. Specifically,~we find the two hash areas that are most similar to each other (based on current $h_{p\in\mathcal{P}^n}$) and conduct sampling on those two states. We use \textit{Jaccard similarity} to measure the similarity between hash codes $A, B$. As such, LSER can always replay the relevant experience.


\subsection{Training Procedure}
We use 
Deep Deterministic Policy Gradient (DDPG)~\cite{lillicrap2015continuous} as the
training backbone. We choose an actor-critic network as the agent and train two parts of the actor-critic network simultaneously.
The critic network aims to minimize the following loss function:
\begin{align*}
    & l(\theta_\psi) = \frac{1}{N} \sum_{j=1}^N ((r + \gamma \xi)-\psi_{\theta_\psi}(s_t,a_t))^2  \\
    & \text{ where } \xi = \psi_{\theta_\psi'}(s_{t+1},\phi_{\theta_\phi'}(s_{t+1}))
\end{align*}
where $\theta_\psi$ 
and
$\theta_\phi$ 
are the critic and actor parameters,
$N$ is the size of 
the
mini-batch from the replay buffer, $\psi_{\theta_\psi'}$ and $\phi_{\theta_\phi'}$ are the target critic and target actor network, respectively.
We apply the Ornstein-Uhlenbeck process in the action space to introduce perturbation; this encourages the agent to explore.
The target network will be updated based on the corresponding hyper-parameter $\tau$.

\section{Experiments}
\begin{algorithm}[!ht]
\SetAlgoLined
 \SetKwInOut{Input}{input}
 \SetKwProg{Fn}{Function}{ is}{end}
 \Input{Transition for storage $\tau:(s_t,a_t,m_t,s_{t+1},r_t)$, capacity $c$, state dimension $d_s$, batch size $b$, Hash bits $n_h$, state $s$ for sampling, epsilon threshold $\epsilon_{max}$.}
  Initialize $n_h$ hyperplanes on projection space $\mathcal{P}_{d_s}$\;
  Initialize empty dictionary $\mathcal{M}$\;
  \Fn{encode(s)}{
    \For{$p$ in $\mathcal{P}_{d_s}$}{
        Calculate hash bits by using Eq.2\;
    }
    \Return{"".join(hash bits)}
  }
  \Fn{Push($\tau$)}{
    $v_h$ = encode($\tau.s_t$) \tcp{get the hash code}
    \eIf{T.size < T.capacity}{
        \eIf{$v_h$ in $\mathcal{M}$}{
            $\mathcal{M}[v_h]$.append($\tau$)\;
        }{
            $\mathcal{M}[v_h]$ = $\tau$ \;
        }
     }
     {
        \If{$v_h$ in $\mathcal{M}$ and $\tau.r_t > \mathcal{M}[v_h][0].r$}{
           \tcp{replace tuple has the minimal reward}
           $\mathcal{M}[v_h][0]$ = $\tau$\;
        }
     }
     Sort($\mathcal{M}[v_h]$) based on reward in ascending order\;
  }
  \Fn{Sample(s,b)}{
    $v_h$ = encode($s$)\;
    $p$ = random.random()\;
    Find two most similar hash values $v_1, v_2$ based on $v_h$\;
    \eIf{$v_h$ in $\mathcal{M}$}{
        \eIf{$p < \epsilon_{max}$}{
            result = $\mathcal{M}[v_h][-b:]$\;
        }{
            result = random.sample($\mathcal{M}[v_h],b$)\;
        }
    }{
        \eIf{$p < \epsilon_{max}$}{
            result = $\mathcal{M}[v_1][-b:] + \mathcal{M}[v_2][-b:]$\;
        }{
            result = random.sample($\mathcal{M}[v_1],b$) + random.sample($\mathcal{M}[v_2],b$)\;
        }
    }
    \Return{result}\;
  } 
 \caption{LSH memory by using dictionary}
 \label{alg:lsh}
\end{algorithm}
\subsection{Online Simulation Platform Evaluation}
We conduct experiments on three widely used public simulation platforms: VirtualTB~\cite{shi2019virtual}, RecSim~\cite{ie2019recsim} and RecoGym~\cite{rohde2018recogym}, which mimic online recommendations in real-world applications.

\vspace{2mm}
\noindent{\bf VirtualTB} is a real-time simulation platform for recommendation, where the agent recommend items based on users' dynamic interests. VirtualTB uses a pre-trained generative adversarial imitation learning (GAIL) to generate different users who have both static interest and dynamic interest. It's worth to mention that, the GAIL is pre-trained by using the real-world from Taobao, which is one of the largest online retail platforms in China. Moreover, the interactions between users and items are generated by GAIL as well. Benefit from that, VirualTB can provide a large number of users and the corresponding interactions to simulate the real-world scenario.

\vspace{2mm}
\noindent{\bf RecSim} is a configurable platform for authoring simulation environments that naturally supports sequential interaction with users in recommender systems. RecSim differs from VirtualTB in containing different, simpler tasks but fewer users and items. There are two different tasks from RecSim, namely interest evolution and long-term satisfaction. The former (interest evolution) encourages the agent to explore and fulfill the user's interest without further exploitation; the latter (long-term satisfaction) depicts an environment where a user interacts with content characterized by the level of `clickbaitiness.' Generally, clickbaity items lead to more engagement yet lower long-term satisfaction, while non-clickbaity items have the opposite effect. The challenge lies in balancing the two to achieve a long-term optimal trade-off under the partially observable dynamics of the system, where satisfaction is a latent variable that can only be inferred from the increase/decrease in engagement.

\vspace{2mm}
\noindent{\bf RecoGym} is a small Open AI gym-based platform, where users have no long-term goals. Different from RecSim and VirtualTB, RecoGym is designed for computational advertising. Similar with RecSim, RecoGym uses the click or not to represent the reward signal. Moreover, similar with RecSim, users in those two environments do not contain any dynamic interests.

Considering RecoGym and RecSim have limited data points and do not consider users' dynamic interests, we select VirtualTB as the main platform for evaluations. Our model is implemented in Pytorch~\cite{paszke2019pytorch} and all experiments are conducted on a server with two Intel Xeon CPU E5-2697 v2 CPUs with 6 NVIDIA TITAN X Pascal GPUs, 2 NVIDIA TITAN RTX and 768 GB memory.
We use two two-hidden-layer neural networks with 128 hidden unit as the actor network and the critic network, respectively. $\tau$, $\gamma$, and $c$ are set to $0.001$, $0.99$ and $1e^6$, respectively, during experiments.

\subsection{Evaluation Metrics and Baselines}
The evaluation metrics are environment-specific.
For VirtualTB and RecoGym, click-through rate is used as the main evaluation metric.
For RecSim, we use the built-in metric, which is a quality score, as the main evaluation metric. We compare our method with the 
following
baselines.
\begin{itemize}
    \item Prioritized Experience Replay (PER)~\cite{schaul2015prioritized}: an experience replay method for discrete control, which uses TD-error to rank experience and a re-weighting method to conduct the bias annealing.
    \item Dynamic Experience Replay (DER)~\cite{luo2020dynamic}: an experience replay method designed for imitation learning, where stores both human demonstrations and previous experience. Those experiences are selected randomly without any priority.
    \item Attentive Experience Replay (AER)~\cite{sun2020attentive}: an experience replay method that
    uses attention to calculate the similarly for boosting sample efficiency with PER.
    \item Selective Experience Replay (SER)~\cite{isele2018selective}: an experience replay method for lifelong machine learning, which employs LSTM as the experience buffer and selectively stores experience.
    \item Hindsight Experience Replay (HER)~\cite{andrychowicz2017hindsight}: an experience replay method that replays two experience (one successful, one unsuccessful) each time. 
\end{itemize}

For AER, PER, SER and HER, We use the same training strategy
as
LSER.
For DER, we use its original structure to run 
experiments without 
human demonstrations.
The size of the replay buffer is set to $1,000,000$ for VirtualTB and $10,000$ for RecSim and RecoGym.
The number of episodes for our experiments is set to $90,000$ for VirtualTB and $1,000$ for RecSim and RecoGym. Note that only PER, AER and SER contains a prioritize operation to rank or store the experience.

\subsection{Results and Evaluation}
\begin{figure*}[ht]
    \centering
    \begin{subfigure}{0.32\linewidth}
    \includegraphics[width=\linewidth]{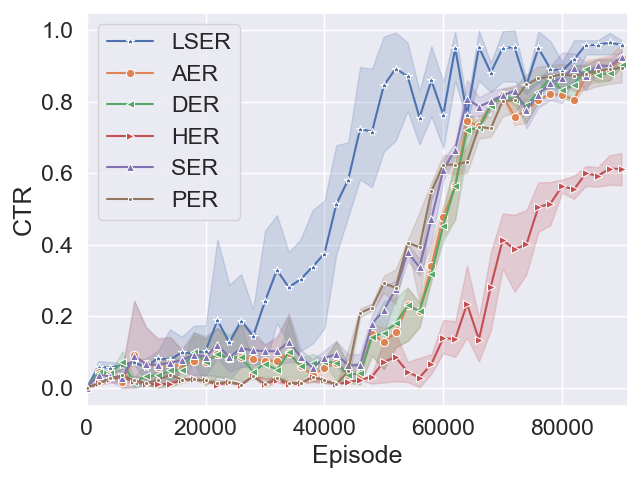}
    \caption{}
    \label{fig:TB}
  \end{subfigure}
  \begin{subfigure}{0.32\linewidth}
    \includegraphics[width=\linewidth]{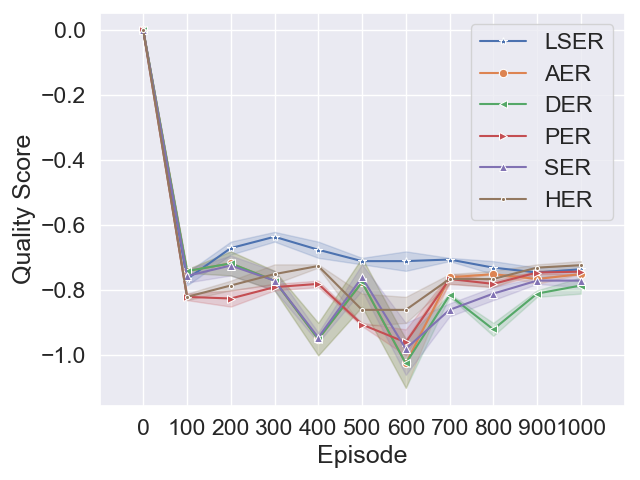}
    \caption{}
     \label{fig:LTS}
  \end{subfigure}
    \begin{subfigure}{0.32\linewidth}
    \includegraphics[width=\linewidth]{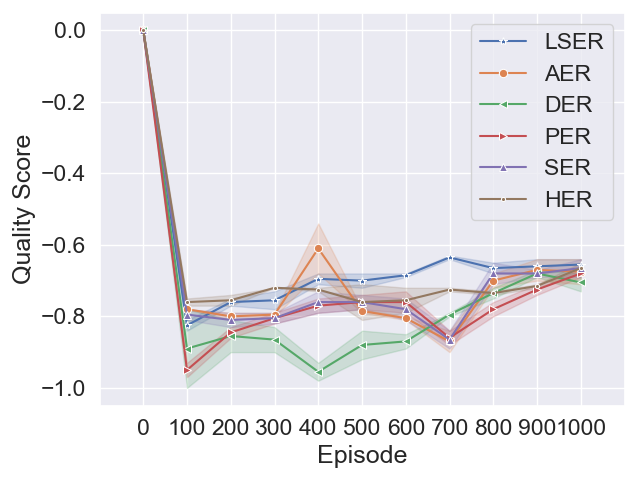}
    \caption{}
    \label{fig:IE}
  \end{subfigure}
  \begin{subfigure}{0.32\linewidth}
    \includegraphics[width=\linewidth]{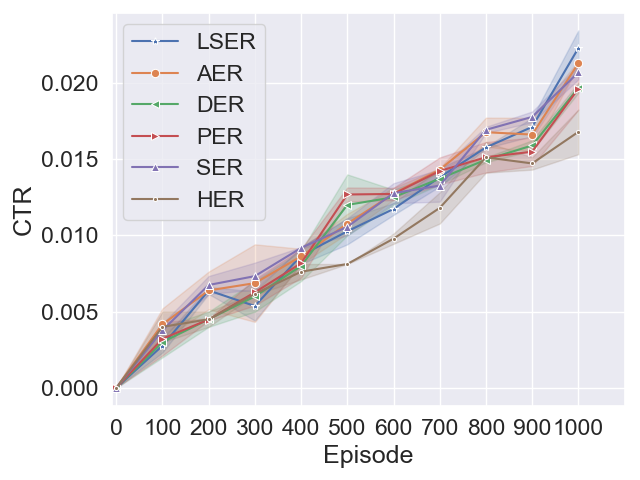}
    \caption{}
    \label{fig:recogym}
  \end{subfigure}
   \begin{subfigure}{0.32\linewidth}
    \includegraphics[width=\linewidth]{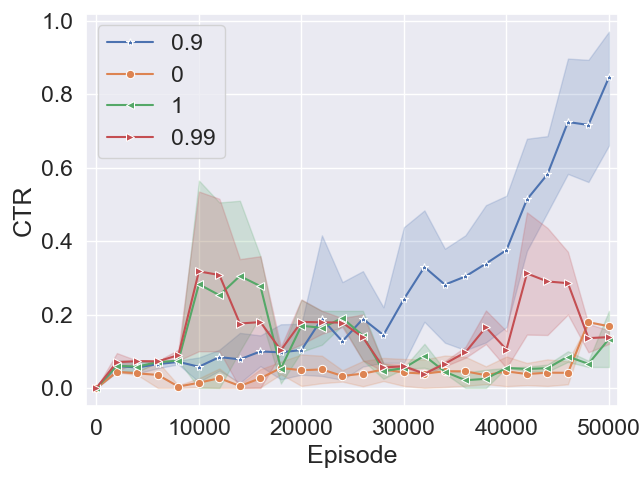}
        \caption{}
    \label{fig:epsilon}
  \end{subfigure}
  \begin{subfigure}{0.32\linewidth}
    \includegraphics[width=\linewidth]{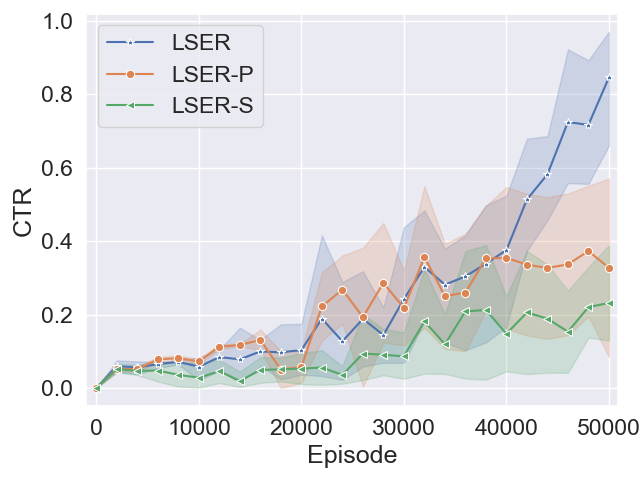}
    \caption{}
    \label{fig:ablation}
  \end{subfigure}
  \caption{Result comparison with four baseline methods on VirtualTB, RecSim and RecoGym. The experiments are repeated five times,
  and mean values are reported.
  95\% confidence intervals are 
  shown.
  (a) is the result for VirtualTB; (b) is the result for long-term satisfaction in RecSim; (c) is the result for interest evolution in RecSim; (d) is the result for RecoGym; (e) is the result for $\epsilon$ study; (f) is the ablation study to show the effectiveness for each component}
  \label{fig:result}
\end{figure*}

Results for the three platforms (Fig~\ref{fig:result}) 
demonstrate our method (LSER) outperformed the baselines: LSER 
yields
significant improvements on VirtualTB, which is a large and sparse environment;
while AER, DER, PER and SER find 
a correct policy within around
$50,000$ episodes, ours
takes around
$30,000$ episodes;
HER does not perform well because 
it introduces too much failed experience and has a slow learning process; 
DER introduces the human demonstration into the vanilla ER which is hard to acquire for recommendation task.


Applying PER to DDPG slightly outperforms applying DER to DDPG, which is consistent with observations by previous work~\cite{novati2019remember,sun2020attentive}.
As PER was originally designed for Deep Q-learning, it uses the high TD-error to indicate the high informative experience for the value-network. When applying PER into DDPG, which is an actor-critic based algorithm, the sampled experience is also used to update the policy network. Those experiences with high TD-error normally diverge far away from the current policy and 
harm 
the updates of 
the
policy-network.
In contrast, LSER selects experience according to the similary with the current state. This preference for on-distribution states tends to discard experiences that contain old states and stabilize the training process of the policy network.
AER does not perform as well as PER in VirtualTB because it heavily relies on the attention mechanism to calculate the similarity score between states. 
LSER's $\epsilon$-greedy method can enforce agent to do more exploration when user's interest shift.

All methods gained similar results on RecSim and RecoGym
because all methods 
can iterate all 
possible combinations of states and actions.
Fig.~\ref{fig:LTS},~\ref{fig:IE}  and~\ref{fig:recogym} show that LSER is slightly better and more stable than 
the baselines on 
RecSim and RecoGym. Since the two platforms are quite small\footnote{RecoGym only contains 100 users and 10 products; RecSim contains 100 users and 20 products.}, similarity matching and $\epsilon$-greedy do not significantly improve 
performance.

\subsection{Running Time Comparison}

We report the running time of 
the
selected experience replay methods in Table~\ref{tab:compare} to evaluate the efficiency of LSER. LSER outperforms 
all 
While performing 
poorly
on RecSim and RecoGym, it is faster than most of the baselines.
In comparison, LSER introduces extra
running time in small environments (e.g, RecSim and RecoGym) than large environments.
For VirutalTB, AER takes much longer time
than all 
other
methods, due to attention 
calculation~\cite{kitaev2020reformer}.

\begin{table}[h]
    \centering
    \caption{Comparison of running time for DER, PER, SER, AER
and LSER coupling with DDPG in three different environments when running the experiments in $90,000$ episodes}
    \begin{tabular}{c|cccc}
    \hline
         & \multicolumn{4}{c}{Running Time ($10^3$ s)}           \\ \cline{2-5} 
         & RecSim(LTS) & RecSim(IE) & RecoGym & VirtualTB \\ \hline
    DER  &   5.63     &  5.42      &      4.53   &     95.22         \\
    PER  &   5.44     &  5.15      &      4.18   &     94.52         \\
    SER  &   5.31     &  5.05      &      4.21   &     \underline{90.05}         \\
    AER  &   \textbf{5.18}     &  \textbf{4.94}      &      \textbf{4.12}  &145.33        \\ 
    HER &    5.33     &  5.11      &      4.20   &     120.33 \\\hline
    LSER &   \underline{5.23}     &  \underline{5.04}      &      \underline{4.15}   &     \textbf{85.12}         \\ \hline
    \end{tabular}
    \label{tab:compare}
\end{table}

\subsection{Ablation Study}
We further investigate the effect of LSER's store and sampling strategy by replacing our store strategy with the normal strategy and our sampling strategy with random sampling.
The results of our ablation study are shown in Fig.~\ref{fig:ablation}, where LSER-P denotes
LSER with the replaced store strategy and LSER-S denotes the LSER with the replaced sampling strategy.
We found the sample strategy played
the
most critical role in achieving good performance, as LSER-S underperformed LSER significantly. 
The store strategy also contributed 
to the better performance. 
LSER-P was less stable (indicated by a wider error bar).
but outperformed LSER at $\sim 30,000$ episodes, due to occurrence of
sub-optimal 
policies.

\subsection{Impact of Number of Hyperplanes}

In our method, the number of hyperplanes is critical to determine the length of the result hash-bits of a given state. Longer hash-bits can provide more accurate similarity measurement result but low efficiency, while shorter hash-bits can increase efficiency but decrease the accuracy. It's a trade-off which needs a middle-point to balance between efficiency and accuracy. We want to answer the following question:\textit{``Does increase the hyperplanes always boost the recommendation performance?''} and find out the optimal number. 

We report the experimental results in VirtualTB, where we evaluate the effect by varying number hyperplanes in LSER (shown in Fig~\ref{fig:hash}). The results on the other two platforms show the similar pattern.
\begin{figure}[h]
    \centering
    \includegraphics[width=0.95\linewidth]{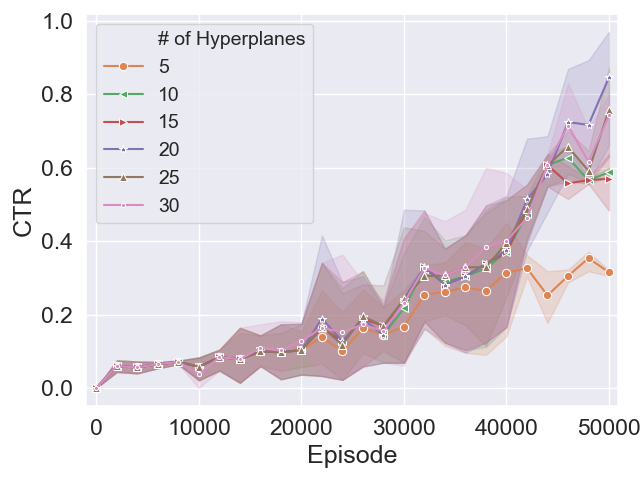}
    \caption{Performance Comparison of different number of hyperplanes}
    \label{fig:hash}
\end{figure}
The performance gradually increases with more hyperplanes, but it levels off or even drops when number of hyperplanes reaches 20.

\subsection{Discussion and Future Extensions}
Fig~\ref{fig:TB} shows LSER suffers instability after reaching the first peak at episode $\sim 50,000$. Different from the other methods, LSER can quickly reach the optimal policy but suffers fluctuation.
That indicates $\epsilon$-greedy tends to lead the agent towards learning from low-priority experience after the optimal policy is reached.
We alleviate the issue by adjusting the value of $\epsilon$.
Here, we tried $\epsilon = \{0, 0.9, 0.99, 1\}$ to determine the best choice of the $\epsilon$ on VirtualTB. The results are shown in Fig.~\ref{fig:epsilon}, where $\epsilon = 1$ corresponds to greedy sampling while $\epsilon=0$ refers to randomly sampling.
Besides, we provide an intervention strategy to 
stabilize the training process---the agent will stop exploration once the test reward is higher than a reward threshold $T_r$. This strategy allows the agent to find an near-optimal policy at an early stage.
We examined the performance under $T_r$=0.95, which delivers a better training process.

\section{Related Work}
\citet{zhao2018deep} first introduced DRL to recommender systems~\citet{zhao2018deep}.
They use DQN to embed user and item information for news recommendaion, where Vanilla ER is used to help 
the
agent learn from 
past experience.
And until present, most
methods use only vanilla ER, which uniformly samples experiences from the replay buffer.
Among them,
~\citet{zhao2020jointly} apply DQN to online 
recommendation and RNN to generate state embeddings; ~\citet{chen2018stabilizing} point out 
that
DQN receives
unstable rewards
in dynamic environments such as online recommendation 
and may harm the
agent;~\citet{chen2019large} 
found
that traditional methods like DQN 
become
intractable when the state 
becomes higher-dimensional; DPG 
addresses the intractability by mapping
high-dimensional discrete state into low-dimensional continuous state~\cite{chen2020knowledge,zhao2020leveraging}.

Intuitively, some 
instances
are more important than others; so a better experience replay 
strategy is to sampling experiences
according to how much current agent can learn from 
each of them.
While such 
a
measure is not
directly accessible, 
proxies propose to retain experiences in the
replay buffer or to sample experiences
from the buffer.
Replay strategies reply on optimization objectives.
In simple continuous control tasks,experience replay should contain 
experiences that are not close to 
the
current policy to prevent fitting to local minima, and the best replay distribution is in between an on-policy distribution and uniform distribution~\cite{de2015importance}. However, 
they~\citet{de2015importance} also note that such a heuristic is unsuitable for complex tasks where policies are updated for many iterations.
In DRL
problems, when the rewards are sparse, 
the
agent can learn from
failed experiences by replacing the original goals with states in
reproduced artificial successful trajectories~\cite{andrychowicz2017hindsight}

For complex control tasks, PER~\cite{schaul2015prioritized} measures the importance of 
experiences using the TD-error and designs a customized 
importance
sampling strategy to avoid the effect of bias. Based on that, Ref-ER~\cite{novati2019remember} actively enforces the similarity between policy and the experience in the replay buffer, considering on-policy transitions are more useful for training the current policy. AER~\cite{sun2020attentive} is an experience replay method that combines the advantages from PER and Ref-ER. It uses attention score to indicate state similarity and replays those experiences awarded high similarity with high priority.
All aforementioned work focuses on optimizing the sampling strategy, aiming to select the salient and relevant agent's experiences in replay buffer effectively. Selective experience replay (SER)~\cite{isele2018selective}, in contrast, aims to optimize the storing process to ensure only valuable experience will be stored. The main idea is to use an Long-short term memory (LSTM) network to store only useful experience.
\section{Conclusion}
In this paper, we propose state-aware reward-driven experience replay (LSER) to address the sub-optimality and training instability issues with reinforcement learning for online recommender systems.
Instead of focusing on improving the sample efficiency for discrete tasks,
LSER considers online recommendation as a continuous task; it then
uses locality-sensitive hashing to determine state similarity and reward for efficient experience replay.
Our 
evaluation of LSER against several state-of-the-art experience-replay methods on three benchmarks (VirtualTB, RecSim and RecoGym) demonstrate LSER's feasibility and superior performance.

In the future, we will explore new solutions for improving stability, such as better optimizers to help the agent get rid of saddle points, new algorithms to stabilize the training for DDPG, and trust region policy optimization to increase training stability~\cite{schulman2015trust}. Moreover, more advance reinforcement learning algorithms could be used to replace the DDPG such as soft actor-critic (SAC)~\cite{haarnoja2018soft} or Twin Delayed Deep Deterministic (TD3)~\cite{fujimoto2018addressing}.
\bibliographystyle{ACM-Reference-Format}
\bibliography{sample-base}

\end{document}